\documentclass[11pt]{article}

\usepackage[a4paper,margin=1in]{geometry}
\usepackage{amsmath,amssymb,amsthm}
\usepackage{hyperref}
\usepackage{enumitem}

\hypersetup{
  hidelinks,
  pdftitle={Quantum Coordination Advantages in AI State-Tracking Tasks:
    Semantic Compilation and Latent Memory},
  pdfauthor={}
}

\newtheorem{definition}{Definition}
\newtheorem{proposition}{Proposition}
\newtheorem{theorem}{Theorem}
\newtheorem{lemma}{Lemma}
\newtheorem{corollary}{Corollary}
\newtheorem{remark}{Remark}

\newcommand{\Cost}{\mathsf{Cost}}

\newcommand{\Tr}{\operatorname{Tr}}
\newcommand{\cl}{\mathrm{cl}}
\newcommand{\q}{\mathrm{q}}
\newcommand{\ket}[1]{\lvert #1\rangle}

\title{Quantum Coordination Advantages in AI State-Tracking Tasks:
Semantic Compilation and Latent Memory}
\author{Ming Yang}
\date{\today}

\begin{document}
\maketitle

\begin{abstract}
We prove inference-time quantum coordination advantages for specified AI
state-tracking tasks.  A solver compresses semantic history into a
future-accessible boundary state and later answers a query.  We count
communication \(B\), persistent instance-dependent memory \(M\), and local
work \(D\); classical recurrence, caches, tools, and recomputation are allowed
and charged.

The central result is a boundary-preserving semantic-compilation theorem.  It
maps a finite one-way, streaming, or adaptive causal task into a semantic AI
interface while preserving event order and access to past input.  Classical
boundary-state lower bounds and quantum-memory upper bounds transfer up to
explicit compiler overhead, independently of the finite-precision recurrent
architecture.

Two applications have classical semantics.  Matched-entity synopsis QA
inherits the hidden-matching separation between \(O(\log N)\) qubits and
\(\Omega(\sqrt N)\) classical boundary bits.  Continual requirements auditing
inherits a Max-\(k\)SAT streaming separation: a recurrent solver uses
\(O(\log^5 n\log(1/\delta))\) qubits and polylogarithmic classical workspace
to obtain a \(0.7172\)-approximation, whereas every classical one-pass
finite-information solver attaining that ratio requires
\(\Omega(\sqrt n)\) coordination width.

As a quantum-native compiler test, a stabilizer latent-state dialogue uses
\(n\) qubits, while every exact finite-state classical causal online
realization satisfies
\(B+M\geq \tfrac12n^2+(\tfrac32-\log_2 3)n+O(1)\).  The source protocols,
streaming algorithms, and stabilizer witness are imported; the new result is
their architecture-independent semantic transfer.  These are memory and
coordination separations, not runtime or empirical advantages for present-day
language models.  The stabilizer result assumes exact simulation and ideal
noiseless quantum memory.
\end{abstract}

\section{Introduction}

Generative modeling is often described as the problem of sampling outputs that
look locally plausible under a learned distribution.  For many tasks this is
not enough.  A model must also preserve an implicit state: the identity of an
object in a story, the value of an unobserved variable, the current branch of a
plan, the state of a simulated world, or the latent rule governing future
observations.  We call this requirement \emph{state tracking} or
\emph{latent-state persistence}.

Recent analyses of transformer architectures identify state tracking as a
structural bottleneck~\cite{huang2025lsp,cui2026topological}.  They
also make clear that classical architectures have real repair mechanisms:
add recurrence, add a state-space memory, write a scratchpad, use tools or
external memory, increase reasoning depth, or run a looped/iterative model.
These repairs are important.  They mean that the right theoretical question is
not
\[
  \text{Can a classical model track state at all?}
\]
but rather
\[
  \text{How much communication, memory, and computation does it need?}
\]

This note develops that second question.  The guiding principle is
boundary-relative: once a computational boundary \(\Sigma\) is chosen, an
explicit signal crossing it is communication, a state retained across it is
memory, and local iterative work performed after crossing it is computation
depth.  The same engineering device may therefore play different roles under
different cuts.  A scratchpad token is just part of the input if it is supplied
by the environment, but it is communication if it was written by an earlier
computational event to coordinate a later one.  A recurrent hidden vector is a
memory resource; an external memory lookup may combine communication and
memory; additional latent thinking steps increase local depth.
This accounting follows the general spacetime-separator framework of
Ref.~\cite{yang2026coordination}, which also distinguishes general classical
causal simulation from restricted global-chart covering.

\paragraph{Relation to quantum generative models.}
Quantum generative models have already been shown to have expressive
advantages over several classical generative model
classes~\cite{gao2022generative}.  The goal here is not to restate that claim.
Instead we ask whether quantum models can reduce the coordination cost of
maintaining a latent state.  For ordinary classical latent variables, there is
no reason to expect a universal quantum advantage: a classical recurrent model
can simply store the variable if enough memory is available.  A sharper target
is a \emph{noncommuting latent state}, where the current query selects a
measurement context and no small classical hidden state can consistently answer
all possible queries.  Contextuality, predictive quantum memory, and
stabilizer simulation lower bounds then become candidates for proving
classical coordination costs.

\paragraph{What would count as an advantage?}
The advantage studied here is not a vague increase in sample quality.  It is a
resource separation.  For a family of target processes \(P_n\), one wants
\[
  \log_2 K_{\cl}^{D,\epsilon}(P_n;\Sigma)
  -
  \log_2 K_{\q}^{\epsilon}(P_n;\Sigma)
\]
to grow with \(n\), or equivalently a lower bound on the classical
coordination bits \(B+M\) that is larger than the number of qubits needed by a
quantum recurrent generator.  A useful result must compare against classical
models that are allowed to use the obvious repairs: recurrence, scratchpads,
external memory, and additional computation.  The streaming-derived results
below are specifically space or retained-information separations rather than
time-complexity separations: the classical update map may perform arbitrary
computation within the one-pass access model, while all information that
survives the next online boundary is charged.  In the general coordination
region, local work is still recorded by \(D\); no result here claims a runtime
speedup.

\paragraph{Contributions and status.}
The paper makes three linked contributions.
\begin{enumerate}[leftmargin=2em]
\item It defines an inference-time quantum-AI state-tracking interface and
places classical architectural repairs inside the \(B,M,D\) accounting of
Ref.~\cite{yang2026coordination}, while keeping general causal simulation
separate from restricted chart and ontological baselines.
\item It proves a boundary-preserving semantic-compilation theorem, together
with one-way and online specializations, that transfers classical space lower
bounds and quantum-memory upper bounds to arbitrary finite-precision
recurrent AI solvers.
\item It instantiates the compiler with matched-entity QA, continual
requirements auditing, and contextual latent-state persistence.  The first
two import communication or streaming separations; the third imports the
finite stabilizer witness of Ref.~\cite{yang2026coordination}.  The new claim
is preservation under the AI interface, not a new hidden-matching, streaming,
or stabilizer lower bound.
\end{enumerate}

\section{Boundary-Relative Coordination Cost}

Let \(P\) be a target generative process.  At time \(t\), the past history is
\[
  h_t=(c_1,o_1,\ldots,c_{t-1},o_{t-1}),
\]
the current condition or query is \(c_t\), and the model outputs \(o_t\).  A
computational boundary \(\Sigma_t\) separates the event that has processed
\(h_t\) from the event that must respond to \(c_t\).

\begin{definition}[Coordination resources]
Relative to \(\Sigma_t\), define:
\[
\begin{array}{rcl}
B &:& \text{explicit information crossing the boundary},\\[1mm]
M &:& \text{internal state retained across the boundary},\\[1mm]
D &:& \text{local processing depth after the boundary}.
\end{array}
\]
The unit of \(B\) and \(M\) is the bit.  For a finite hidden-state space
\(\Lambda\), storing a state requires \(\lceil \log_2|\Lambda|\rceil\) bits.
\end{definition}

\begin{definition}[Generative coordination region]
For a model class \(\mathcal A\), error tolerance \(\epsilon\), and boundary
choice \(\Sigma\), define
\[
\begin{aligned}
\Cost_{\mathcal A}^{\epsilon}(P;\Sigma)=\{(B,M,D):{}&
\text{some }\mathcal A\text{-generator simulates }P\\
&\text{within error }\epsilon\text{ using }(B,M,D)\}.
\end{aligned}
\]
The distance of this region from the origin measures the coordination burden
of the process under the chosen boundary.
\end{definition}

\begin{remark}
The boundary is part of the model.  The prompt or context is not automatically
communication.  It counts as communication only when it carries information
from one modeled computational event to another.  This convention prevents the
resource accounting from prejudging whether a long context, a scratchpad, a
KV cache, or a recurrent state is the relevant resource.
\end{remark}

\subsection{General and restricted classical baselines}

The term \emph{classical simulation} has two different strengths that must not
be conflated~\cite{yang2026coordination}.  A \emph{general classical causal
simulator} may use randomized, context-dependent response kernels and adaptive
state updates.  Its complete future-accessible boundary configuration
\(\lambda\) must contain all past-dependent information available after the
cut, but the transition and response rules themselves are unrestricted.
For a static past--future probability table
\[
  F_\Sigma(P)_{u,v}=\Pr_P(U=u,V=v),
\]
the general separator theorem gives
\[
  B+M\geq\log_2\operatorname{rank}_+F_\Sigma(P).
\]
For an online process with Hankel table \(H_P\), the corresponding quantity is
the causal positive-realization rank: one common family of response and update
kernels must realize every history and continuation.  At depth \(D\),
\[
  B+M
  \geq
  \log_2\operatorname{rank}_+^{{\rm causal},D}(H_P).
\]
These are general within the stated causal boundary and access model.

A \emph{restricted chart simulator} instead requires every counted
transcript--memory state to select a context-independent global response chart.
Its covering-number bound
\[
  B+M\geq\log_2\chi_G^D
\]
applies only to that cover-admissible class.  A KWB-compatible stabilizer
simulator is a different restricted sequential model: it imposes exact state
update and single-shot-distinguishability support conditions.  Neither a chart
covering lower bound nor a KWB support-counting lower bound applies to a general
causal simulator without an explicit reduction.  For the exact stabilizer
family used below, Ref.~\cite{yang2026coordination} supplies that reduction by
constructing a finite adaptive witness whose partitioning and distinguishing
tests force the KWB overlap count on every finite-state causal realization.

Accordingly, \(K_{\cl}^{D,\epsilon}\) below is always relative to a named
classical model class.  When a concrete theorem uses an unrestricted one-way or
one-pass baseline, it says so.  When only a chart, HMM, or other restricted lower
bound is available, the restriction is part of the claim and must not be
silently promoted to a bound on arbitrary classical AI.

\section{Classical Latent-State Generators}

A finite classical latent-state generator has hidden state
\(\lambda_t\in\Lambda\).  The past history prepares a distribution over hidden
states,
\[
  w_t(\lambda)=w(\lambda\mid h_t),
\]
and the current query \(c_t\) is answered by a response kernel
\[
  R(o_t\mid c_t,\lambda_t).
\]
Thus
\[
  q(o_t\mid h_t,c_t)
  =
  \sum_{\lambda_t\in\Lambda}
  w(\lambda_t\mid h_t)R(o_t\mid c_t,\lambda_t).
\]
After observing \(c_t,o_t\), the state may update by another stochastic kernel
\[
  U(\lambda_{t+1}\mid \lambda_t,c_t,o_t).
\]

\begin{definition}[Classical state complexity]
Let \(K_{\cl}^{D,\epsilon}(P;\Sigma)\) be the minimum number of classical
hidden coordination states needed by a depth-\(D\) classical generator to
simulate \(P\) within error \(\epsilon\) relative to \(\Sigma\).  The associated
memory cost is
\[
  M_{\cl}^{D,\epsilon}(P;\Sigma)
  =
  \left\lceil \log_2 K_{\cl}^{D,\epsilon}(P;\Sigma)\right\rceil .
\]
\end{definition}

\begin{definition}[Quantum state complexity]
Let \(K_{\q}^{\epsilon}(P;\Sigma)\) be the minimum Hilbert-space dimension
needed by a quantum latent-state generator to simulate \(P\) within error
\(\epsilon\) relative to \(\Sigma\).  The associated quantum memory cost is
\[
  Q^{\epsilon}(P;\Sigma)
  =
  \left\lceil \log_2 K_{\q}^{\epsilon}(P;\Sigma)\right\rceil
\]
qubits.
\end{definition}

The classical model class is an argument suppressed by the notation.  Thus
\(K_{\cl}^{D,\epsilon}\) may denote the general causal class, an HMM class, or
a more restricted simulator class only when that choice is stated locally.

\begin{proposition}[Basic state-count lower bound]
Suppose a simulator with \(B\) bits of explicit boundary communication and
\(M\) bits of retained memory can select among at most \(2^{B+M}\) effective
classical coordination states in a depth-\(D\) model class.  Then exact or
\(\epsilon\)-approximate simulation of \(P\) implies
\[
  B+M
  \geq
  \log_2 K_{\cl}^{D,\epsilon}(P;\Sigma).
\]
\end{proposition}

\begin{proof}
The transcript and retained memory jointly take at most \(2^{B+M}\) values.
If each value determines one effective state of the simulator class, then the
simulator can realize no more than \(2^{B+M}\) effective coordination states.
By minimality of \(K_{\cl}^{D,\epsilon}(P;\Sigma)\), this number must be at
least \(K_{\cl}^{D,\epsilon}(P;\Sigma)\).
\end{proof}

\begin{proposition}[Ordinary latent-state repair]
\label{prop:ordinary-repair}
Suppose \(P\) has a finite classical sufficient state \(z_t\in Z\) such that
\[
  P(o_t\mid h_t,c_t)=P(o_t\mid z_t,c_t)
\]
and \(z_{t+1}\) is sampled from a kernel depending only on
\((z_t,c_t,o_t)\).  Then, for a model class that can implement the required
kernels at depth \(D\),
\[
  K_{\cl}^{D,0}(P;\Sigma)\le |Z|,
  \qquad
  M_{\cl}^{D,0}(P;\Sigma)\le \lceil \log_2|Z|\rceil .
\]
\end{proposition}

\begin{proof}
Use the sufficient state itself as the hidden coordination state
\(\lambda_t=z_t\).  The response and update kernels are precisely the kernels
given in the hypothesis.
\end{proof}

\begin{remark}
This elementary upper bound is the reason that transformer failures alone do
not prove a quantum advantage.  A recurrent classical model may repair the
failure by spending memory.  The relevant question is whether the required
classical memory, communication, or depth grows asymptotically faster than the
quantum latent state needed for the same process.
\end{remark}

\section{Classical Repairs as Resource Moves}

This section records how common classical repairs to state-tracking failures
fit the resource region.

\begin{enumerate}[leftmargin=2em]
\item \textbf{Recurrence and state-space models.}  Recurrent neural networks,
state-space models, and recurrent transformers carry a hidden state across
time.  In the present accounting, they primarily increase \(M\).
\item \textbf{Scratchpads and chain-of-thought traces.}  When an earlier
generation step writes information for a later step to read, the trace is
explicit boundary traffic.  It is therefore a \(B\)-type resource under that
time cut.
\item \textbf{External memory and tools.}  Tool calls and retrieval systems
mix communication and memory: the main generator sends a query, receives a
message, and may store a summary.
\item \textbf{Latent thinking and iterative inference.}  Extra internal
iterations increase \(D\), the local processing depth available after the
current condition is known.
\item \textbf{Long context.}  Long context is not automatically a cost.  If it
is supplied as part of the task input, it is data.  If it is generated or
compressed by the model to coordinate future events, it is part of the
resource accounting.
\end{enumerate}

Thus classical solutions are not dismissed.  They are the baselines whose
resource costs the theory tries to quantify.

\section{Quantum and Noncommuting Latent States}

A quantum latent-state generator carries a density operator
\[
  \rho_t\in\mathcal D(\mathcal H_t).
\]
For each current condition \(c_t\), the generator uses a POVM
\[
  \{M_{o}^{c_t}\}_o
\]
and outputs
\[
  p(o_t\mid h_t,c_t)
  =
  \Tr(M_{o_t}^{c_t}\rho_t).
\]
The latent state then updates by a quantum instrument,
\[
  \rho_{t+1}
  =
  \mathcal E_{c_t,o_t}(\rho_t).
\]

The important distinction is not merely that \(\rho_t\) is continuous or
random.  It is that different queries \(c_t\) may correspond to noncommuting
measurements.  A classical hidden state can answer all queries by carrying a
large enough table of counterfactual responses.  Contextuality lower bounds
ask how large that table, or the coordination mechanism selecting among
tables, must be.

\begin{definition}[Noncommuting latent-state task]
A latent-state tracking task is noncommuting if there are query contexts
\(c,c'\) whose associated measurements cannot be jointly represented as
coarse-grainings of one fixed classical response variable without increasing
the hidden coordination state space.  In finite measurement scenarios this
condition can be formalized by the absence, or high cost, of a
noncontextual/global-chart representation of the induced empirical model.
\end{definition}

\begin{definition}[Induced one-step empirical model]
Fix a history class \(H\) and a set of allowed queries \(\mathcal C\).  The
one-step empirical model induced by \(P\) is the family
\[
  E_h=\{P(\cdot\mid h,c):c\in\mathcal C\}_{h\in H}.
\]
A classical hidden-state representation with \(K\) states writes every member
as
\[
  P(o\mid h,c)
  \approx
  \sum_{\lambda=1}^{K} w_h(\lambda)R(o\mid c,\lambda).
\]
If the same \(\lambda\) must answer several incompatible queries \(c\), the
problem reduces to finding small classical charts for the query family.
\end{definition}

\begin{proposition}[Contextual coordination lower bound]
\label{prop:contextual-lower-bound}
Let \(P\) induce a finite query model \(E_h\), and fix a classical model
class \(\mathcal C\).  If every depth-\(D\) representation in \(\mathcal C\)
within error \(\epsilon\) needs at least \(K\) effective states, then any
depth-\(D\) implementation in \(\mathcal C\) with boundary resources
\((B,M)\) must satisfy
\[
  B+M\ge \log_2 K .
\]
\end{proposition}

\begin{proof}
This is the basic state-count lower bound applied to the induced query model.
Each boundary transcript and retained state selects one effective classical
response state for the query family.  Fewer than \(K\) such states cannot
realize the required empirical model within the specified error.
\end{proof}

\begin{remark}
Ordinary natural-language ambiguity or hidden world state is not by itself
evidence for quantum advantage.  The quantum advantage target is narrower:
latent states whose observable queries have a contextual or noncommuting
structure, or classical stochastic processes whose predictive states admit a
provably smaller quantum representation.
\end{remark}

\section{Separation Criterion}

The preceding definitions give a simple separation template.

\begin{theorem}[Coordination-cost separation criterion]
Let \(P_n\) be a family of generative processes and \(\mathcal C_n\) a named
classical simulator class.  Suppose:
\begin{enumerate}[leftmargin=2em]
\item there is a quantum latent-state generator for \(P_n\) with Hilbert-space
dimension \(d_n\), and
\item every depth-\(D_n\) generator in \(\mathcal C_n\) simulating \(P_n\)
within error \(\epsilon_n\) needs at least \(L_n\) hidden coordination states.
\end{enumerate}
Then every such implementation in \(\mathcal C_n\) satisfies
\[
  B+M\ge \log_2 L_n,
\]
whereas the quantum implementation uses at most
\[
  \lceil \log_2 d_n\rceil
\]
qubits of latent memory.  The separation is linear, polynomial, or exponential
according to the growth of
\[
  \log_2 L_n-\log_2 d_n .
\]
\end{theorem}

\begin{proof}
The classical bound is Proposition~\ref{prop:contextual-lower-bound}.  The
quantum upper bound is the definition of Hilbert-space dimension as quantum
memory.  The final statement is only a naming convention for the asymptotic
gap.
\end{proof}

This theorem is deliberately formal.  It identifies exactly what has to be
proved in any proposed application: a quantum generator of small dimension and
a lower bound on classical coordination states against the chosen classical
baseline.

\section{Related Work: Expressivity, Memory, and State Tracking}

Prior work on quantum generative models shows that quantum correlations can
increase the expressive power of Born machines, quantum circuit generators,
and related models~\cite{gao2022generative}.  Recent work has also emphasized
the need for task-level metrics when assessing whether such expressive
advantages matter for practical learning problems~\cite{gili2024metrics}.
Quantum predictive-memory
separations similarly show that quantum states can reduce the memory needed
to simulate certain stochastic processes~\cite{gu2012occam,garner2017memory}.
That line of work is about representing or sampling distributions.  The
present note asks a different but related question: whether maintaining a
latent state over time can be done with less coordination.

An \emph{expressivity advantage} asks which distributions can be represented
compactly.  A \emph{coordination advantage} asks how many states, messages,
or update steps are needed to preserve one latent process.
The two can interact.  A distributional separation may imply a
state-complexity separation for a one-shot process, while a recurrent
contextual process may produce an advantage even when each local output
distribution is simple.

\paragraph{Contextuality as classical memory cost.}
Sequential contextuality was already formulated as a classical internal-memory
cost by Kleinmann et al., and Fagundes and Kleinmann extended that analysis to
the full probabilistic Peres--Mermin correlations
\cite{kleinmann2011memory,fagundes2017memory}.  Karanjai, Wallman, and Bartlett
later obtained growing stabilizer-simulation memory bounds
\cite{karanjai2018contextuality}.  Most directly, Prakash converts
graph-theoretic contextuality into an exponential quantum-memory advantage for
a formal-language promise problem, including bounded-error probabilistic
automata and entropic bounds~\cite{prakash2026memory}.  That result uses a
quantum finite automaton and an exclusivity-sensitive classical automaton
model; allowing finite confusability changes its exponential conclusion.  It
is direct prior art for contextuality-to-memory advantage and is not subsumed
here.  Conversely, the present theorem concerns a generic
boundary-preserving compiler across named one-way, streaming, and adaptive
causal access models; it does not reproduce Prakash's bounded-error automata
result.

\paragraph{AI state tracking and architectural lower bounds.}
Language-model state tracking has also been studied through learned
permutation-composition mechanisms, expressivity limits for state-space
models, and communication-complexity limitations of transformer layers
\cite{li2025state,merrill2024illusion,peng2024limitations}.  These results
motivate the semantic interfaces used below, but they diagnose particular
architectures or learned mechanisms.  The present resource comparison instead
charges the complete future-accessible state of any solver in the named access
model and asks whether quantum memory changes that cost.

\paragraph{Relation to Gao et al. and strong \(k\)-contextuality.}
Gao et al. show that
quantum correlations can give compact generative representations outside the
reach of selected classical Bayesian-network and neural-network families, and
they explicitly connect their separations to nonlocality and contextuality
\cite{gao2022generative}.  Their hidden-Markov-model result is especially
close: a basis-enhanced 2-gram model with state-space dimension \(D\) cannot,
under their support/KL criterion, be represented by a translation-form
classical HMM with fewer than \(D^{\Omega(\log D)}\) hidden units.  Since the
logarithm of the hidden-state count is a memory cost, this already has the
shape
\[
  \text{classical memory}=\Omega((\log D)^2),
  \qquad
  \text{quantum memory}=O(\log D),
\]
for that simulator class.

Teo et al. give a still more direct precursor to the contextual memory
interpretation~\cite{teo2025kcontextuality}.  They define strong
\(k\)-contextuality for translation tasks and show that a strongly
\(k\)-contextual task cannot be represented to finite relative entropy by a
classical streaming model with fewer than \(k\) latent states.  They also give
algorithms for estimating the relevant contextuality quantity and study it as
an empirical heuristic for memory separation.  Their theorem already
establishes a contextuality-to-classical-memory link; the present paper does
not claim that link as new.

The present note should therefore not be read as discovering that quantum
correlations can help generative models.  Rather, it recasts such separations
as boundary-relative coordination statements and connects them to modern
state-tracking failures.  Relative to Gao et al. and Teo et al., the added
claims are the boundary-preserving semantic transfer across one-way,
streaming, and adaptive causal access models, together with explicit
\(B,M,D\) accounting across recurrent AI repairs.  If a classical
model repairs a deficit by carrying a larger recurrent state, writing a
scratchpad, using external memory, or recomputing from the transcript, those
repairs are allowed but charged to \(M\), \(B\), or \(D\).  In this language,
Gao et al.'s HMM separation and Teo et al.'s strong \(k\)-contextuality theorem
are precursor coordination-cost statements; the contextual-LSP and dialogue
formulations below ask for corresponding statements under an explicit
interactive boundary.

\section{Boundary-Preserving Semantic Compilation}
\label{sec:semantic-compilation}

The source of a quantum advantage may be a communication problem, a streaming
algorithm, or an adaptive quantum process.  To make it an AI state-tracking
task, one must encode its events as semantic inputs and its answers as
task-level outputs without changing what information is available at each
time.  The following definition isolates that requirement.

\begin{definition}[Boundary-preserving semantic compiler]
Let \(\Pi_n\) be a finite task with an ordered sequence of environment events,
solver actions, and computational boundaries specified by an access model
\(\mathcal A_n\).  A boundary-preserving semantic compiler maps each source
event online to a finite text, symbolic, or multimodal block and maps solver
outputs back to the source output alphabet.  It must:
\begin{enumerate}[leftmargin=2em]
\item preserve the event order, adaptive choices, and source boundaries;
\item never re-supply a past source event after its boundary unless that
record is explicitly charged as persistent state;
\item preserve the source acceptance relation or transcript distribution up
to error \(\eta_n\); and
\item use at most \(a_n\) bits of parser, renderer, and compiler workspace,
including every compiler record retained across a source boundary.
\end{enumerate}
The resulting semantic task is denoted
\(\mathsf{Sem}_{\mathcal A_n}(\Pi_n)\).
\end{definition}

Let \(C_{\cl}^{\mathcal A,\epsilon}(\Pi_n)\) denote the minimum peak number of
classical bits in the complete future-accessible boundary state of a
finite-state causal solver for \(\Pi_n\) with error at most \(\epsilon\).
A quantum resource profile \((q_n,c_n)\) means \(q_n\) retained qubits and
\(c_n\) retained classical bits under the same access model.

\begin{theorem}[Semantic coordination transfer]
\label{thm:semantic-coordination-transfer}
Let \(\mathsf{Sem}_{\mathcal A_n}\) be a boundary-preserving semantic compiler
with workspace \(a_n\), compilation error \(\eta_n\), and at most \(T_n\)
event boundaries.  Then:
\begin{enumerate}[leftmargin=2em]
\item every classical solver for
\(\mathsf{Sem}_{\mathcal A_n}(\Pi_n)\) with error at most \(\epsilon\) and
peak coordination width \(W_\Sigma\) satisfies
\[
 W_\Sigma
 \geq
 C_{\cl}^{\mathcal A,\epsilon+\eta_n}(\Pi_n)
 -a_n-O(\log T_n);
\]
\item if \(\Pi_n\) has a quantum solver with error at most \(\epsilon\) and
resource profile \((q_n,c_n)\), then the compiled semantic task has a solver
with error at most \(\epsilon+\eta_n\) using \(q_n\) qubits and
\[
 c_n+a_n+O(\log T_n)
\]
classical bits across the corresponding boundaries.
\end{enumerate}
For an exact eventwise compiler that retains no additional
instance-dependent state, exact classical causal-state lower bounds and
quantum-memory upper bounds are preserved without asymptotic loss.
\end{theorem}

\begin{proof}
Compose a classical semantic solver with the online encoder, parser, and
output decoder.  Because the compiler preserves event order and does not
reintroduce expired source records, the composition is a valid
\(\mathcal A_n\)-solver for \(\Pi_n\).  At each source boundary its complete
state consists of the semantic solver's state, at most \(a_n\) compiler bits,
and \(O(\log T_n)\) event-counter bits.  Its error is at most
\(\epsilon+\eta_n\).  The definition of
\(C_{\cl}^{\mathcal A,\epsilon+\eta_n}\) gives the first inequality.

For the quantum direction, run the semantic parser online, apply the source
quantum channel or measurement selected by the decoded event, retain its
\(q_n\)-qubit state and \(c_n\)-bit classical state, and render the source
output semantically.  The workspace and error overheads are those stated.
When the compiler is exact and retains no state, the compositions preserve
the source boundary-state sets themselves, giving the final claim.
\end{proof}

This theorem is the paper's central transfer principle.  It does not create a
new quantum communication protocol, streaming algorithm, or contextuality
lower bound.  It states when any such result becomes an inference-time quantum
AI result rather than a superficial relabelling: the semantic interface must
preserve the causal boundary on which the resource comparison is made.

\section{AI State-Tracking Baselines and One-Way Applications}

Transformer state-tracking failures provide the motivating classical
bottleneck~\cite{vaswani2017attention,huang2025lsp,cui2026topological}.
A feed-forward transformer can often use the context as a workaround, but
persistent dynamic state is not free: the state must be stored, rewritten,
copied into a scratchpad, retrieved from memory, or recomputed.  The
coordination-cost picture turns those engineering choices into resource moves:
\[
  \text{recurrence}\mapsto M,\qquad
  \text{scratchpad or tool messages}\mapsto B,\qquad
  \text{extra latent computation}\mapsto D.
\]

This perspective treats classical repairs as serious baselines.  A quantum
proposal must beat not only a plain transformer, but also classical models
that are allowed to spend the relevant \(B,M,D\) resources.

\paragraph{Latent-state persistence benchmarks.}
The latent-state-persistence tasks of Huang et al. are useful because they
separate local linguistic plausibility from the ability to preserve a hidden
state across many queries~\cite{huang2025lsp}.  In this paper they play the
role of a classical stress test, not a quantum benchmark.  Number guessing,
yes/no state tracking, and related tasks can usually be repaired by giving the
classical model a sufficient recurrent state, a scratchpad, or external
memory.  The contribution of the present framework is to charge those repairs:
the recurrent state contributes to \(M\), the scratchpad contributes to \(B\),
and repeated inference contributes to \(D\).  To obtain a specifically quantum
advantage, one must enrich this kind of benchmark so that the hidden state is
not merely unknown but query-contextual or noncommuting.

\paragraph{Topological state-tracking dialogues.}
The state-tracking problems emphasized by Mozer, Siddiqui, and Liu provide a
second useful shell~\cite{cui2026topological}.  A transcript describes a path
of local updates, and later prompts query the current state.  If the state
space is an ordinary classical space, this again admits an ordinary recurrent
repair: store a sufficient coordinate for the current state.  The route to a
quantum separation is therefore not to claim that these classical tasks are
already quantum, but to keep their dialogue structure while replacing the
tracked state by a noncommuting latent state.  The history then prepares
\(\rho_h\), the prompt selects a measurement context \(c\), and the answer is
sampled from \(\Tr(M^c_o\rho_h)\).  In this form, the topological/state-
tracking burden is preserved, but the classical repair is no longer a small
coordinate; it is a contextual response chart whose size can be lower-bounded.

\paragraph{Relational state tracking.}
Multi-entity text tasks provide a natural source of large coordination costs.
If a history determines \(N=\Theta(n^2)\) independent pairwise facts among
\(n\) entities and a query asks for any selected fact, then an exact classical
tracker needs \(\Theta(n^2)\) retained bits.  However, this observation alone
does not imply a quantum advantage.  An \(m\)-qubit latent representation that
can answer any one of the \(N\) independent facts with success probability
\(p>1/2\) is a quantum random-access code, so Nayak's bound gives
\[
  m\ge (1-H_2(p))N=\Omega(n^2),
\]
where \(H_2\) is the binary entropy~\cite{nayak1999randomaccess}.  Thus an
arbitrary classical relation table cannot be compressed to \(O(n)\) qubits if
the benchmark allows reliable random access to all its entries.

The plausible quantum target is narrower: the text-induced IO relation should
have high classical coordination rank but low quantum, or positive
semidefinite, rank.  In binary-output form, a family of histories and queries
defines a nonnegative matrix
\[
  A_{h,c}=P(o=1\mid h,c).
\]
A classical latent-state factorization corresponds to a nonnegative
factorization of \(A\), while a quantum latent-state representation has the
form
\[
  A_{h,c}=\Tr(E_c\rho_h),
\]
which is a positive-semidefinite factorization~\cite{fawzi2015psdrank}.  A
natural-text or text-wrapped relational benchmark would therefore show a
coordination advantage only if its conditional-response matrix has large
nonnegative rank but small PSD rank.  Stabilizer state tracking is one
structured instance of this pattern: it is not an arbitrary table of
classical pairwise facts, but a noncommuting relational state whose query
responses are compactly represented by an \(n\)-qubit state.

A more classical-looking route comes from one-way communication complexity.
Here it serves as a calibration of the same boundary resource, stated at the
level of a general reading-comprehension schema rather than a single
hand-picked relation.

\subsection{One-way communication calibration}

\begin{definition}[Entity-attribute synopsis QA]
An entity-attribute synopsis QA family consists of finite sets
\(\mathcal X_N\) of passage states, \(\mathcal C_N\) of query contexts, and
\(\mathcal O_N\) of answers, together with a relation
\[
  R_N\subseteq \mathcal X_N\times \mathcal C_N\times \mathcal O_N .
\]
A passage \(h_x\) is an unambiguous text encoding of an entity-attribute
state \(x\in\mathcal X_N\).  A query \(c\in\mathcal C_N\) is revealed after
the passage has been processed.  A valid answer is any \(o\in\mathcal O_N\)
such that \((x,c,o)\in R_N\).
\end{definition}

The benchmark boundary is essential.  The passage is read first, then only a
boundary state is retained, and the query is revealed later.  If the full
passage is carried across the boundary, its length is charged to \(B\); if a
classical model rereads or rescans the passage after seeing the query, that
repair is charged to \(D\).

\begin{corollary}[One-way lift to synopsis QA]
\label{thm:one-way-lift}
Suppose the relation problem \(R_N\), under a chosen input distribution,
admits a one-way quantum protocol with \(q_N\) qubits and success probability
at least \(1-\epsilon\), while every one-way randomized classical protocol
with the same success probability requires at least \(c_N\) bits.  Then the
corresponding entity-attribute synopsis QA family has a quantum
state-tracking solver using \(q_N\) qubits across the passage--query boundary,
and every classical solver in the same one-way boundary model satisfies
\[
  B+M\ge c_N .
\]
\end{corollary}

\begin{proof}
This is the exact one-boundary specialization of
Theorem~\ref{thm:semantic-coordination-transfer}, with the passage and query
as the two source events and no retained compiler state.
The quantum upper bound is obtained by running the quantum one-way encoder
after parsing the passage and retaining its \(q_N\)-qubit message as the
boundary state.  After the query is revealed, the query responder runs the
one-way decoder and renders its output as text.

Conversely, any classical QA solver using \(m=B+M\) boundary bits gives a
one-way classical protocol for \(R_N\): Alice parses \(x\), forms the passage
\(h_x\), runs the passage processor, and sends the resulting boundary state
to Bob; Bob parses \(c\), runs the query responder, and outputs its answer.
Thus \(m\ge c_N\).
\end{proof}

The content is therefore not tied to a specific Boolean operation.  Direct
random access to arbitrary entity attributes is ruled out by the quantum
random-access-code obstruction above, but any entity-attribute relation family
with a one-way quantum/classical separation yields a synopsis-QA separation.
Hidden matching is a simple instantiation of this more general lifting
principle.

In that instantiation, a passage describes \(N\) named records, each with a
binary attribute such as cohort, stance, access level, or case label.  After
the passage has been processed, a later query gives a list of disjoint record
pairs and asks the model to report any listed pair together with whether the
two records have the same or different labels.  This is a normal
database-style question about entities and relations; it does not mention
quantum physics.

Under the memory version of this task, the problem is exactly the hidden
matching problem: it has an \(O(\log N)\)-qubit one-way protocol and requires
\(\Omega(\sqrt N)\) classical one-way bits at bounded error
\cite{baryossef2004hiddenmatching}.  Kerenidis and Raz study the related
Boolean Hidden Matching partial function~\cite{kerenidis2006booleanhiddenmatching};
the relational result used here is the Bar-Yossef--Jayram--Kerenidis problem.
This
avoids the random-access-code obstruction because the query does not ask for a
pre-specified stored bit; it lets the solver choose any edge from a large
matching and report the corresponding relation.

\paragraph{Chart interpretation.}
This interpretation uses the sheaf-theoretic framework for contextuality and
its database reading~\cite{abramsky2011sheaf,abramsky2012databases}.  The relevant
``global'' notion here is the local-to-global one: local tables or contexts
are easy to satisfy, while the issue is whether they can be glued into a
single global section.  We do not impose the physical ``genuine global''
conditions of multipartite KS scenarios.

For a fixed matching \(M\), the response condition is local: output one edge
of \(M\) and the corresponding parity.  A classical boundary state \(s\),
however, induces a global response chart
\[
  g_s:\; M\longmapsto (i,j,b)
\]
over all possible matchings.  A solver using \(B+M\) classical bits can select
at most \(2^{B+M}\) such charts after reading the history.  The hidden-matching
lower bound says that, at bounded error, no small family of classical global
charts can cover the required history--query relation.

This is not a bare KS contradiction.  If the full string \(x\) is stored,
then the assignment
\[
  b_{ij}=x_i\oplus x_j
\]
is a perfectly good global parity chart for all pairs.  The obstruction is
resource-sensitive: classical simulation must spend many bits to select an
adequate chart, while the quantum protocol keeps a compact phase state from
which a query context extracts one valid local relation by interference.  In
this sense, hidden matching is a classical-looking local-to-global task whose
coordination advantage can be read in a resource-sensitive global-chart
language.  The local-to-global language is not new; the added point is to
charge the number of selectable charts to the boundary resources of a
state-tracking generator.

\begin{definition}[Matched-entity consistency QA]
Let \(N\) be even.  A passage \(h_x\) describes \(N\) named entities with
binary labels \(x\in\{0,1\}^N\), using a fixed unambiguous grammar.  A query
\(c_M\) presents a perfect matching \(M\) on the entity set \([N]\).  A valid
answer is any triple \((i,j,b)\) such that \((i,j)\in M\) and
\[
  b=x_i\oplus x_j .
\]
Here \(b=0\) means that the two selected entities have the same label, and
\(b=1\) means that they have different labels.
\end{definition}

\begin{corollary}[Matched-entity QA separation]
\label{thm:matched-entity-qa}
Consider the one-way state-tracking protocol in which \(x\) is drawn
uniformly from \(\{0,1\}^N\), the passage processor sees \(h_x\), a boundary
state is retained, and only then a uniformly random perfect matching \(M\) is
revealed to the query responder.  The query responder must output a valid
triple for \((x,M)\) with probability at least \(2/3\).

There is an exact quantum boundary-state protocol using
\(\lceil\log_2 N\rceil\) qubits.  Any bounded-error classical protocol in the
same one-way boundary model requires
\[
  B+M=\Omega(\sqrt N)
\]
bits.
\end{corollary}

\begin{proof}
For the quantum upper bound, after reading the passage prepare
\[
  \ket{\psi_x}
  =
  \frac{1}{\sqrt N}
  \sum_{i=1}^{N}(-1)^{x_i}\ket{i},
\]
which uses \(\lceil\log_2 N\rceil\) qubits.  Given a matching \(M\), measure
first in the orthogonal decomposition
\[
  \mathrm{span}\{\ket{i},\ket{j}\},\qquad (i,j)\in M,
\]
which selects an edge \((i,j)\in M\).  Conditional on this edge, the state is
proportional to
\[
  (-1)^{x_i}\ket{i}+(-1)^{x_j}\ket{j}.
\]
Now measure in the basis
\[
  \frac{\ket{i}\pm\ket{j}}{\sqrt 2}
\]
inside that two-dimensional subspace.  The sign is \(+\) iff
\(x_i\oplus x_j=0\) and \(-\) iff \(x_i\oplus x_j=1\), so the responder outputs
\((i,j,x_i\oplus x_j)\) with certainty.

For the classical lower bound, suppose a classical state-tracking solver uses
\(m=B+M\) boundary bits and succeeds with probability at least \(2/3\).  This
solver gives a one-way randomized communication protocol for hidden matching:
Alice, given \(x\), forms the passage \(h_x\), runs the passage processor, and
sends the resulting \(m\)-bit boundary state to Bob; Bob, given \(M\), runs
the query responder and outputs its triple.  The success probability is the
same as that of the QA solver.  The one-way communication lower bound for
hidden matching therefore implies \(m=\Omega(\sqrt N)\)
\cite{baryossef2004hiddenmatching}.
\end{proof}

\paragraph{Why RNNs and state-space models do not trivialize the question.}
They do trivialize one weak claim: it is not enough to show that a fixed-depth
feed-forward transformer loses track of a latent variable.  A recurrent model
can store the variable.  But recurrence changes the resource point from small
\(M\) to larger \(M\).  The nontrivial question is whether, for some process
family \(P_n\), every classical recurrent repair requires
\(\Omega(f(n))\) bits while a quantum latent state uses \(O(g(n))\) qubits with
\(g(n)\ll f(n)\).

\section{Online Quantum-AI Specialization and Coordination Width}

The one-way synopsis theorem places one boundary between a completed passage
and a later query.  State tracking in the sense of an update rule
\(s_t=f(s_{t-1},x_t)\) instead places a boundary after every update.  This
section gives the corresponding online notion and shows how ordinary
streaming-space lower bounds become architecture-independent coordination
lower bounds.  It is the one-pass specialization of
Theorem~\ref{thm:semantic-coordination-transfer}.

\begin{definition}[Peak online coordination width]
Consider a solver that consumes update blocks
\(x_1,\ldots,x_T\) in order.  Let \(Z_t\) contain all stream-dependent
information available after \(x_t\) has been consumed and before
\(x_{t+1}\) arrives.  This includes retained activations, recurrent states,
accessible cache entries, generated scratchpad symbols, and records written to
an external tool or store.  If \(Z_t\) has at most \(2^{w_t}\)
operationally distinguishable classical values, define
\[
  W_\Sigma=\max_{0\leq t\leq T}w_t
\]
to be the solver's peak online coordination width across the family of cuts
\(\Sigma=(\Sigma_0,\ldots,\Sigma_T)\).
\end{definition}

Under the bit accounting used above,
\[
  W_\Sigma\leq\max_t(B_t+M_t)
\]
when every accessible explicit record and internal state is included in
\(B_t+M_t\).  Conversely, representing \(Z_t\) by an index costs at most
\(w_t\) bits.  Thus \(W_\Sigma\) is the sequential, peak-space projection of
the coordination region.  A neural state with \(r\) real coordinates at
\(p\)-bit operational precision contributes at most \(rp\) bits; allowing an
exact real number to encode an unbounded stream would leave the finite-space
model and is not a finite-information classical baseline.

\begin{definition}[Semantics-preserving online compiler]
Let \(\Pi_n\) be a streaming relation problem with update alphabet
\(\mathcal U_n\) and output relation
\[
  R_{\Pi_n}\subseteq \mathcal U_n^*\times\mathcal Y_n.
\]
A semantics-preserving online compiler consists of a prefix-decodable encoding
\(\mathsf{Enc}_n(u)\) of each update as one text or multimodal block and an
answer decoder \(\mathsf{Dec}_n\) such that
\[
  (u_{1:T},y)\in R_{\Pi_n}
  \quad\Longleftrightarrow\quad
  \bigl(\mathsf{Enc}_n(u_1),\ldots,
        \mathsf{Enc}_n(u_T),\mathsf{Dec}_n^{-1}(y)\bigr)
  \text{ is accepted}.
\]
The compiler has overhead \(a_n\) if parsing the current block and rendering
the final answer use at most \(a_n\) bits of workspace and retain no
additional stream-dependent information between update boundaries.
\end{definition}

The no-retained-information clause prevents the linguistic wrapper itself from
hiding a large database.  It does not require constant-length text: entity
identifiers may use \(O(\log n)\) bits, provided only the current record is
being parsed.

\begin{corollary}[Online semantic lift]
\label{thm:online-semantic-lift}
Let \(S_{\cl}(n,\epsilon)\) be a lower bound on the space of every randomized
one-pass classical streaming algorithm for \(\Pi_n\) with error at most
\(\epsilon\).  If an online compiled AI solver has error at most \(\epsilon\),
compiler overhead \(a_n\), and peak classical coordination width
\(W_\Sigma\), then
\[
  W_\Sigma
  \geq
  S_{\cl}(n,\epsilon)-a_n-O(\log T).
\]
If \(\Pi_n\) has a one-pass quantum streaming algorithm using \(S_{\q}\)
qubits and \(C_{\q}\) classical bits, then the compiled task has a quantum
recurrent solver using
\[
  S_{\q}\ \text{qubits}
  \qquad\text{and}\qquad
  C_{\q}+a_n+O(\log T)\ \text{classical bits}.
\]
\end{corollary}

\begin{proof}
Use Theorem~\ref{thm:semantic-coordination-transfer} with the one-pass access
model and the exact eventwise compiler above.  Substituting
\(C_{\cl}^{\mathcal A,\epsilon}=S_{\cl}(n,\epsilon)\) gives the classical
inequality, while the source profile \((S_{\q},C_{\q})\) gives the displayed
quantum and classical workspace bounds.
\end{proof}

\begin{remark}[Architecture independence]
The classical implication uses only the number of distinguishable states
carried across update cuts.  The update map may be nonlinear, randomized, and
computationally unbounded.  It therefore applies equally to finite-precision
RNNs, nonlinear SSMs, recurrent transformers, KV-cache systems, scratchpads,
and tool-using agents, provided all persistent information is counted in
\(W_\Sigma\).  Extra local depth cannot reconstruct distinctions that were not
retained after the stream passed.
\end{remark}

\begin{remark}[Fixed parameters versus instance-dependent state]
The parameters of a pretrained model are part of the fixed algorithm
description and are not charged as online memory.  Write a recurrent
implementation schematically as
\[
  z_{t+1}=F_\theta(z_t,x_t),
  \qquad
  p(o_t\mid c_t,z_t)=G_\theta(c_t,z_t).
\]
The streaming lower bound already permits \(F_\theta\) and \(G_\theta\) to be
arbitrarily complicated.  Nevertheless, if two realized histories induce the
same future-accessible state \(z_t\), fixed parameters cannot make their
response distributions differ under the same future query.  Model weights may
store the update rule or a vast read-only lookup table, but the
instance-dependent index selecting the realized history must still cross the
boundary.  Test-time weight updates, adapters, fast weights, or model selection
that depend on the stream are therefore part of \(Z_t\) and are charged to
\(W_\Sigma\).
\end{remark}

\begin{remark}[The full-context loophole]
If the complete raw transcript remains freely available for random access,
the solver is no longer one-pass and Corollary~\ref{thm:online-semantic-lift}
does not apply.  One must either charge the stored transcript as external
memory and its retrieval as boundary traffic, or analyze a multi-pass model.
This is precisely the distinction between explicit dynamic state and the
transformer workaround of re-examining its whole context.  The access regime,
not the linguistic surface, is therefore part of the theorem statement.
An \(L\)-token context over a vocabulary of size \(V\) can itself carry up to
\(L\log_2V\) raw token-index bits, and its accessible KV cache is also
stream-dependent state.  A sufficiently large context can therefore satisfy
the lower bounds in this paper; it is a classical repair with a potentially
large \(W_\Sigma\), not a violation of the theorem.
\end{remark}

\subsection{Technical graph instance: dynamic relation summaries}

The first calibration is a direct technical instance of the lift.  It has an
ordinary graph and database interpretation.
There are \(n\) named entities.  Each update block states one directed
relation, for example, ``entity \(u\) places a one-way dependency on entity
\(v\).''  Once a block has been processed, it is not supplied again.  At the
end the solver must estimate the maximum number of reported relations that can
run from one side of a bipartition to the other.  Formally, removing the fixed
grammar leaves the insertion-only Max-DiCut edge stream.

\begin{corollary}[Dynamic relation-summary separation]
\label{prop:dynamic-relation-summary}
For the terminal approximation ratio \(0.4844\) and failure probability
\(\delta\), the dynamic relation-summary task has a quantum recurrent solver
using
\[
  O\!\left(\log^5 n\log\frac1\delta\right)
\]
qubits of online workspace, plus logarithmic compiler workspace.  Every
finite-information classical recurrent solver in the same one-pass access
regime satisfies
\[
  W_\Sigma=\Omega(\sqrt n).
\]
Equivalently, its family of effective recurrent coordination states has size
\[
  K_{\cl}^{\rm online}\geq 2^{\Omega(\sqrt n)}.
\]
\end{corollary}

\begin{proof}
Kallaugher, Parekh, and Voronova give a one-pass quantum streaming algorithm
with the displayed space bound and approximation ratio.  The classical
streaming lower bound they invoke states that every ratio strictly larger than
\(4/9\) requires \(\Omega(\sqrt n)\) bits~\cite{kpv2024maxdicut}.
The fixed relation grammar is prefix-decodable with \(O(\log n)\) workspace.
Corollary~\ref{thm:online-semantic-lift} transfers both bounds, and
\(0.4844>4/9\).  Exponentiating the width lower bound gives the state-count
form.
\end{proof}

This task realizes DeepMind's schematic update
\(s_t=f(s_{t-1},x_t)\): each sentence modifies a compact synopsis of a
growing relational world.  The conclusion is stronger than a failure theorem
for a feed-forward transformer.  Giving the model recurrence repairs the
topological depth problem, but every classical repair still needs
\(\Omega(\sqrt n)\) peak retained bits at the target approximation ratio.

\subsection{Natural AI task: continual requirements auditing}

Consider an AI assistant supporting a long-running policy, planning, or
engineering process.  The participants introduce requirements one at a time;
the assistant must update its synopsis without retaining or rereading the
whole transcript.  At the end, the user asks for the best-achievable compliance
score: how many of the accumulated requirements can any coherent plan satisfy?

The binary decisions may represent whether to activate a service, approve a
proposal, allocate a team, or adopt a design option.  Typical utterances are:
``If weekend hours are not extended, remote triage must be enabled,'' ``Either
the mobile unit stays onsite or weekend hours are extended,'' and ``Remote
triage and moving the mobile unit offsite may not occur together.''  After
semantic parsing, these become disjunctions of signed binary decisions.  The
user-facing problem is requirements auditing; its formal semantic core is
Max-\(k\)SAT.

\begin{definition}[Continual requirements-audit task]
Fix \(k\geq2\).  A task instance contains \(n\) named binary decisions and a
time-ordered dialogue \(r_1,\ldots,r_T\).  Each requirement utterance \(r_t\)
has a certified semantic parse as a clause \(C_t\) containing at most \(k\)
literals.  Once \(r_t\) has been processed, it is unavailable except through
the solver's retained state.  On the terminal query, the solver outputs a
number \(Z\) estimating
\[
  \operatorname{OPT}(C_{1:T})
  =
  \max_{a\in\{0,1\}^n}
  \bigl|\{t:C_t(a)=1\}\bigr|,
\]
or equivalently the normalized compliance score
\(\operatorname{OPT}(C_{1:T})/T\).
\end{definition}

The theorem-certified version uses a controlled natural-language grammar, so
each requirement can be parsed independently with \(O(\log n)\) workspace.
A benchmark may additionally contain ordinary paraphrases, domain vocabulary,
and coreference, but then semantic-parser error is a separate empirical layer.
The memory theorem already applies to the exactly parseable subset; a quantum
upper bound for the richer surface form additionally assumes a shared online
semantic front end.

\begin{corollary}[Continual requirements-audit separation]
\label{prop:requirements-audit}
For every fixed \(k\geq2\), the controlled-language continual
requirements-audit task admits a one-pass quantum recurrent solver which, with
probability at least \(1-\delta\), outputs \(Z\) satisfying
\[
  \operatorname{OPT}(C_{1:T})
  \geq Z
  \geq 0.7172\,\operatorname{OPT}(C_{1:T})
\]
using
\[
  O\!\left(\log^5 n\log\frac1\delta\right)
\]
qubits of online workspace.  It also uses polylogarithmic classical working
bits, including \(O(\log n)\) exact counters.  Every finite-information
classical recurrent solver attaining that ratio in the same one-pass regime has
\[
  W_\Sigma=\Omega(\sqrt n).
\]
\end{corollary}

\begin{proof}
Wang and Yang give the displayed one-pass quantum streaming algorithm for
Max-\(k\)SAT; its quantum sketches also use polylogarithmic classical control
and working bits, and its preprocessing retains logarithmic exact counters.
The classical streaming lower bound rules out every ratio
strictly larger than \(\sqrt2/2\approx0.7071\) in
\(o(\sqrt n)\) space~\cite{wang2026maxksat}.  The controlled requirement
grammar is a semantics-preserving online compiler with logarithmic workspace,
so Corollary~\ref{thm:online-semantic-lift} preserves both bounds.
\end{proof}

This is a natural AI state-tracking problem in the operational sense used by
Mozer, Siddiqui, and Liu: the accumulated requirement set is an evolving world
state, and its task-sufficient synopsis must be updated as
\(s_t=f(s_{t-1},r_t)\).  The assistant is not asked to recall arbitrary past
sentences or output a quantum object.  It produces one classical planning
diagnostic.  The result is stronger than the observation that a transformer
may lose track of a satisfying assignment: the lower bound ranges over all
bounded-space classical update rules, including recurrent repairs.

The task also has a clear limitation.  It estimates the optimum compliance
value; it does not output the optimizing plan.  The cited quantum streaming
algorithm does not establish a compact quantum advantage for plan
construction, and the present paper does not claim one.

\paragraph{What is imported and what is new.}
The Max-DiCut and Max-\(k\)SAT quantum algorithms, approximation constants,
and classical streaming lower bounds are imported results.  Rewording their
records as sentences does not create a new quantum algorithm.  The new claim
developed here is the general transfer principle and the associated AI task
model: after fixing an online semantic boundary, a streaming lower bound
becomes a lower bound on the peak coordination width of every
finite-information recurrent AI implementation, including the standard
architectural repairs to transformer state tracking.  Continual requirements
auditing supplies a practical planning semantics for that theorem.  A
purported small classical solver must be using uncharged transcript access,
unbounded numerical precision, a weaker output guarantee, or a different
access model.

\paragraph{Finite-size interpretation.}
These theorems establish asymptotic coordination separations, not a practical
memory saving at ordinary LLM scales.  The \(\Omega(\sqrt n)\) bounds hide
constants: at \(n=10^6\), the scaling term \(\sqrt n\) is only \(10^3\),
whereas a 128k-token context over a \(10^5\)-word vocabulary can carry about
\(2.1\times10^6\) raw token-index bits before counting the physical KV cache.
Likewise, the explicit stabilizer expression below is about
\(5.0\times10^3\) bits
at \(n=100\) qubits.  Moreover, an \(O(\log^5 n)\) quantum upper bound need
not beat \(\sqrt n\) at moderate \(n\), especially after constants,
fault-tolerance, and interface costs are included.  No finite-size crossover
or practical quantum-memory advantage is claimed here.

\section{Contextual Quantum-AI Application}
\label{sec:contextual-benchmark}

The previous sections explain what kind of benchmark is needed.  Ordinary LSP
tasks hide a classical variable and ask later questions about it.  Such tasks
are useful probes of state tracking, but a classical recurrent model can repair
them by storing a sufficient state.  To obtain a candidate quantum separation,
the hidden state should instead be queried through incompatible contexts.

\begin{definition}[Contextual LSP benchmark]
For each size parameter \(n\), a contextual latent-state-persistence benchmark
consists of:
\[
\begin{array}{rcl}
\mathcal H_n &:& \text{allowed histories},\\
\mathcal C_n &:& \text{allowed query contexts},\\
\{\rho_h:h\in\mathcal H_n\} &:& \text{latent states prepared by histories},\\
\{M^c_o:o\in O_c\}_{c\in\mathcal C_n}
&:& \text{query-dependent output measurements},\\
\{\mathcal E_{c,o}\}_{c,o}
&:& \text{state-update instruments}.
\end{array}
\]
At test time the benchmark presents a history \(h_t\) and a query context
\(c_t\).  The target conditional distribution is
\[
  P_n(o_t\mid h_t,c_t)
  =
  \Tr(M^{c_t}_{o_t}\rho_{h_t}),
\]
and after observing \(o_t\) the latent state updates as
\[
  \rho_{h_{t+1}}
  =
  \mathcal E_{c_t,o_t}(\rho_{h_t}).
\]
A model is evaluated by the average total-variation distance, log loss, or
success probability of its conditional predictions over an adaptive sequence
of histories and queries.
\end{definition}

The intended boundary is the time cut after \(h_t\) has been processed but
before \(c_t\) is revealed.  If the complete history is re-supplied together
with the query, then a classical model may recompute the latent state from the
raw transcript; in the present accounting that repair is charged to local
depth \(D\), not treated as free state tracking.

This definition contains ordinary LSP as the jointly classical special case.
If all states and measurements are jointly diagonal in a common basis, then
there is a classical sufficient variable \(z_t\) and
Proposition~\ref{prop:ordinary-repair} applies.
The task becomes contextual in the operational sense used here only when the
same history can later be queried in contexts that do not admit a small common
response chart.

\paragraph{From LSP to contextual LSP.}
The operational modification is small.  In an ordinary hidden-state benchmark,
the history prepares a latent variable and later questions ask for facts about
that variable.  In a contextual benchmark, the history prepares a latent
object and later questions choose one of several incompatible tests of that
object.  A language wrapper could describe the history as a lab notebook,
simulation trace, symbolic circuit, or world-state update; the mathematical
core is that \(c_t\) is not just a request for a stored fact, but a measurement
context.

The resulting classical repair options are still allowed.  A classical model
may store a chart in recurrent memory, write intermediate chart data into a
scratchpad, or recompute a chart after seeing the query.  The point is that
these repairs now have visible costs:
\[
  \text{stored chart data}\mapsto M,\qquad
  \text{written chart data}\mapsto B,\qquad
  \text{reconstructed chart data}\mapsto D.
\]

\begin{definition}[Stabilizer contextual LSP]
The stabilizer contextual-LSP family \(P^{\mathrm{stab}}_n\), the benchmark
version of the quantum-memory seed imported from
Ref.~\cite{yang2026coordination}, is obtained by taking
\(\mathcal H_n\) to be histories of Clifford gates and previous Pauli
measurement outcomes on \(n\) qubits.  Each history prepares an \(n\)-qubit
stabilizer state \(\rho_h\).  A query context \(c\in\mathcal C_n\) is a
commuting family of Pauli observables, the output \(o\) is the corresponding
string of measurement outcomes, and the update map is the usual stabilizer
measurement update.
\end{definition}

\begin{definition}[Semantic stabilizer dialogue]
A semantic stabilizer dialogue is a natural-language or symbolic-language
presentation of \(P^{\mathrm{stab}}_n\).  The transcript describes Clifford
updates and previous Pauli measurement outcomes using an unambiguous finite
grammar; the next prompt describes a commuting Pauli context; and the required
answer is the corresponding outcome distribution or a sample from it.  The
semantic target process is still \(P^{\mathrm{stab}}_n\); the text wrapper
only supplies a state-tracking interface of the kind used in transformer
state-tracking benchmarks.
\end{definition}

\begin{definition}[Adaptive-complete recurrent simulation]
A classical recurrent implementation of a contextual-LSP family is
\emph{adaptive-complete} if, for every finite adaptive policy that chooses the
next query context as a function of the previous history and outcomes, the
implementation reproduces the joint distribution of the full transcript.  The
boundary state at time \(t\) consists of the retained recurrent state together
with any explicit transcript crossing the chosen boundary \(\Sigma_t\).  Thus
an implementation using resources \((B,M)\) has at most \(2^{B+M}\) effective
boundary states at each cut.
\end{definition}

\paragraph{Relation to the genuine-global construction.}
Reference~\cite{yang2026coordination} contains two logically distinct steps.
Its finite-causal-witness lemma is already a single-system statement about the
\(n\)-qubit stabilizer seed: it constructs the finite adaptive interface
\(W_n\) and proves the causal-state lower bound used below.  A subsequent,
optional flag lift embeds that seed into a genuinely global multipartite model
and transfers the same cost by conditioning on the flag.  The present LSP
benchmark imports only the first step.  It therefore needs no \(k=1\) or
single-party reduction from the genuinely global theorem, and its quantum
upper bound remains the \(n\)-qubit seed realization.  Applying the separate
flag lift would instead produce a genuinely global restriction with a small
additional flag-memory overhead, but that extra physical structure is not
used in the AI state-tracking claim.

\begin{lemma}[Imported finite causal stabilizer witness]
\label{prop:finite-causal-witness}
For every \(n\geq2\), there is a finite adaptive interface
\(W_n\subset P_n^{\mathrm{stab}}\) such that every exact finite-state
classical causal online realization of \(W_n\) has at least
\[
  K_n
  \geq
  \frac{2^n\prod_{j=1}^{n}(2^j+1)}
       {5\cdot3^{n-2}}
\]
boundary states.  Consequently, every exact adaptive-complete classical
recurrent implementation of \(P_n^{\mathrm{stab}}\) satisfies
\[
  B+M\geq\log_2K_n .
\]
\end{lemma}

\begin{proof}
Karanjai--Wallman--Bartlett show that every set of more than
\(m_n=5\cdot3^{n-2}\) pure \(n\)-qubit stabilizer states admits a stabilizer
partitioning measurement~\cite{karanjai2018contextuality}.  For every subset
of \(m_n+1\) preparations, include one such measurement and, after each
relevant outcome, one allowed single-shot test distinguishing the resulting
orthogonal pair.  The stabilizer preparation and measurement sets are finite
for fixed \(n\), so their union defines a finite interface \(W_n\).

If one causal boundary state occurred with positive probability after all
preparations in one of these subsets, its common response kernel would assign
positive probability to some partitioning outcome and successor state.  Two
preparations would then reach orthogonal postmeasurement records through that
same successor state, while the subsequent distinguishing test requires
different certain outcomes, a contradiction.  Thus one causal state can occur
in the support of at most \(m_n\) pure preparations.  There are
\(2^n\prod_{j=1}^{n}(2^j+1)\) pure stabilizer states, giving the stated
count.  This finite-witness upgrade from the KWB overlap theorem to arbitrary
finite-state causal realizations is exactly the single-system
``finite causal witness from the stabilizer overlap bound'' lemma of
Ref.~\cite{yang2026coordination}, restated here to make the reduction
self-contained.  It precedes, and does not rely on, that reference's
genuinely global flag-lift theorem.

An exact adaptive-complete implementation of the full process remains exact
when restricted to \(W_n\).  Its complete future-accessible transcript and
retained state have at most \(2^{B+M}\) values, so
\(2^{B+M}\geq K_n\).
\end{proof}

\begin{corollary}[Imported stabilizer latent-state separation]
\label{thm:stabilizer-lsp}
The family \(P^{\mathrm{stab}}_n\) is exactly generated by a quantum recurrent
model using \(n\) qubits of latent memory.  Any exact adaptive-complete
finite-state classical causal recurrent implementation, with all
future-accessible boundary information counted, satisfies
\[
  B+M=\Omega(n^2).
\]
More explicitly, for \(n\ge 2\),
\[
  B+M
  \ge
  \log_2
  \left(
    \frac{2^n\prod_{j=1}^{n}(2^j+1)}
         {5\cdot 3^{n-2}}
  \right)
  =
  \frac12n^2+
  \left(\frac32-\log_2 3\right)n+O(1)
  =
  \Omega(n^2).
\]
\end{corollary}

\begin{proof}
The quantum implementation stores the physical \(n\)-qubit stabilizer state
and applies the requested Clifford or Pauli-measurement update.  Restricting an
exact adaptive-complete classical implementation to the finite interface
\(W_n\) preserves its boundary-state set.  Lemma
\ref{prop:finite-causal-witness} therefore gives the displayed ratio and its
logarithm.  Expanding
\(\log_2(2^j+1)=j+\log_2(1+2^{-j})\) gives the quadratic asymptotic.
\end{proof}

The factor in this count is \(2^j+1\), exponential in \(j\).  As a
low-dimensional check, the formula gives \(2^3(3)(5)(9)=1080\) pure
stabilizer states at \(n=3\); this factor is what produces the quadratic
logarithmic growth.

\begin{corollary}[Contextual quantum-AI state-tracking separation]
\label{prop:text-wrapped-stabilizer}
Relative to the boundary after the transcript has been semantically processed
and before the next query context is revealed, the text-wrapped stabilizer
dialogue is generated by a quantum recurrent model with \(n\) qubits of
latent memory.  Any exact adaptive-complete finite-state classical causal
recurrent solver, with all future-accessible semantic boundary information
counted, satisfies
\[
  B+M=\Omega(n^2).
\]
\end{corollary}

\begin{proof}
The controlled grammar is an exact boundary-preserving semantic compiler:
removing the linguistic surface leaves the stabilizer contextual-LSP process
\(P^{\mathrm{stab}}_n\), each current block is parsed with \(O(\log n)\)
workspace, and no parser record survives the semantic boundary.  A quantum
solver applies each decoded update to the physical \(n\)-qubit latent state.
Theorem~\ref{thm:semantic-coordination-transfer}, applied to the imported
finite-state lower bound in
Lemma~\ref{prop:finite-causal-witness}, gives the stated
\(\Omega(n^2)\) classical bound.
\end{proof}

\begin{remark}
This corollary is the sense in which an NLP-style task can carry a provable
quantum coordination advantage.  The advantage belongs to the noncommuting
state-tracking problem preserved by the text, not to ordinary natural
language understanding by itself.  If the full raw transcript is re-supplied
and a classical model recomputes a stabilizer tableau after each prompt, the
stored or retransmitted transcript is charged to \(B+M\) and its processing to
\(D\); neither resource is free.
\end{remark}

\begin{remark}
This benchmark is intentionally quantum-native.  Rephrasing it as natural
language does not make the advantage a property of ordinary language modeling;
it only gives a user-facing wrapper around a noncommuting state-tracking task.
Its value is that it isolates the exact technical target: prove the reduction
from an adaptive state-tracking interface to causal boundary-state complexity.
The finite witness above achieves that reduction exactly; robust approximate
and natural-task versions remain open.
\end{remark}

\paragraph{Scope of the transferred bound.}
The displayed \(\Omega(n^2)\) separation applies to every exact finite-state
classical causal online realization of the adaptive-complete interface; it is
not restricted to a preselected chart or neural architecture.  The finite
witness is what upgrades the KWB overlap count to this general causal class.
The theorem still does not cover constant-error approximation, a batch
algorithm with free random access to the full transcript, or uncounted
infinite-precision state.  Those are different approximation or access models.

\paragraph{Why adaptive completeness matters.}
A one-step benchmark distribution is weaker than a process simulator.  A
model might predict the next answer well on a fixed distribution of histories
without carrying enough state to answer all future compatible queries.  The
stabilizer lower-bound route therefore needs an adaptive benchmark: after any
history that the model itself has helped generate, an evaluator may choose a
new commuting Pauli context and continue.  This is the operational content of
state tracking, and it is what turns a conditional prediction benchmark into a
candidate memory lower-bound problem.

\paragraph{Evaluation protocol.}
An evaluator can implement adaptive completeness without inspecting the
model's internal state.  At the start of each trial it resets the model,
supplies an allowed preparation-and-update history, and then selects each next
commuting Pauli context as a function of the full interaction transcript.  It
records the model's response, applies the corresponding target-process update,
and continues for a prescribed horizon.  Repeated trials estimate the joint
distribution of complete adaptive trajectories, which is compared with the
target process in total variation or log loss.  The policy family must include
context choices that distinguish histories merged by a candidate simulator;
a fixed i.i.d. test set does not provide this guarantee.  Constructing an
efficient worst-case policy, or a finite certificate that is complete for a
given model class, remains an open algorithmic problem.

\paragraph{Classical-baseline audit.}
The logical strength of each result is summarized below.  Here
\emph{general} always means general within the displayed one-way, one-pass, or
causal access model, not a batch algorithm with free access to the complete
input.
\begin{center}
\small
\begin{tabular}{p{0.24\textwidth}p{0.43\textwidth}p{0.19\textwidth}}
\hline
Result & Classical baseline & Status \\
\hline
State-count and separation criteria
& The explicitly named class \(\mathcal C\)
& Conditional on \(\mathcal C\) \\
Semantic coordination transfer
& The named source access model, with all compiler state charged
& General transfer \\
Synopsis and matched-entity QA
& Arbitrary randomized one-way protocol
& General one-way \\
Online semantic lift and its two tasks
& Arbitrary randomized one-pass finite-information update rule
& General one-pass \\
Strong \(k\)-contextuality of Teo et al.
& Finite-state HMM or finite-precision autoregressive realization
& Interface-limited \\
Imported stabilizer witness
& Arbitrary exact finite-state causal online realization
& Ref.~\cite{yang2026coordination}, exact \\
Compiled contextual AI dialogue
& Arbitrary exact finite-state causal online realization
& General causal transfer \\
\hline
\end{tabular}
\end{center}

\section{Discussion: Calibrations and Extensions}

The task families above play three different roles and should not be read as
coequal headline contributions.

\begin{enumerate}[leftmargin=2em]
\item \textbf{Classical-looking calibrations.}  Matched-entity QA and continual
requirements auditing instantiate established one-way or streaming
separations inside ordinary database and planning semantics.  They validate
the boundary accounting and the semantic lifts, but their underlying quantum
algorithms and lower bounds are imported.  In particular, the
requirements-audit advantage comes from streaming Max-\(k\)SAT and is not
claimed to arise from contextuality.
\item \textbf{Quantum-native compiler test.}  Stabilizer measurement tracking
imports the clean finite witness of Ref.~\cite{yang2026coordination}.  The new
role it plays here is to test whether a semantic state-tracking interface
preserves an adaptive causal-state separation.  Its limitations are explicit:
adaptive completeness is required, the source theorem is exact, and the
semantics remain quantum-native unless a certified wrapper is used.
\item \textbf{Extensions toward natural processes.}  Quantum
predictive-memory models provide classical stochastic processes with compact
quantum representations~\cite{gu2012occam,garner2017memory}.  Contextual
interactive dialogues, scientific time series, and embodied-agent histories
are further targets, but difficult state tracking alone does not imply
noncommutativity.  A convincing extension must identify a concrete query
family with a large classical causal-state lower bound and a compact quantum
realization.
\end{enumerate}

\section{Open Problems}

\begin{enumerate}[leftmargin=2em]
\item \textbf{Foundational robustness.}  Identify equivalence classes of
computational boundaries under which the \(B,M,D\) region is stable, and
separate representation or simulation cost from the cost of learning a
generator.  Without these distinctions, a lower bound may partly reflect
bookkeeping or training assumptions rather than an intrinsic coordination
obstruction.
\item \textbf{Robust and efficient witnesses.}  Extend the exact stabilizer
theorem to \(\epsilon\)-approximate causal simulation and replace the
potentially large finite witness by efficient adaptive evaluator policies and
statistical certificates.  Semantic wrappers must preserve the same complete
causal interface, including gates, measurements, randomness, and updates.
\item \textbf{Finite and physical advantage.}  Determine constants, crossover
scales, matching classical upper bounds, and actual context or KV-cache costs.
A physical comparison must also include coherence time, error correction,
refresh, and classical--quantum input/output overhead over the full stream.
\item \textbf{Natural tasks and outputs.}  Find noncommuting latent-state
processes closer to practical generation while retaining provable causal-state
lower bounds.  This includes classical-looking interactive benchmarks,
natural-language variants with controlled parser error, and strengthening the
requirements-audit result from value estimation to construction of an
approximately optimal plan.
\end{enumerate}

\section*{Acknowledgements}

The author acknowledges the use of AI-assisted tools during the preparation of
this draft for literature exploration, mathematical checking, language
polishing, and organization.  The author is solely responsible for all claims,
interpretations, and conclusions.


\begin{thebibliography}{99}

\bibitem{huang2025lsp}
J.-t. Huang, K. Sun, W. Wang, and M. Dredze.
\newblock On the failure of latent state persistence in large language models.
\newblock arXiv:2505.10571, 2025.

\bibitem{cui2026topological}
M. C. Mozer, S. A. Siddiqui, and R. Liu.
\newblock The topological trouble with transformers.
\newblock arXiv:2604.17121, 2026.

\bibitem{yang2026coordination}
M. Yang.
\newblock Genuine global Kochen--Specker contextuality as classical
coordination cost.
\newblock arXiv:2606.23577, 2026.

\bibitem{gao2022generative}
X. Gao, E. R. Anschuetz, S.-T. Wang, J. I. Cirac, and M. D. Lukin.
\newblock Enhancing generative models via quantum correlations.
\newblock \emph{Physical Review X} 12, 021037 (2022).

\bibitem{gili2024metrics}
K. Gili, M. Mauri, and A. Perdomo-Ortiz.
\newblock Generalization metrics for practical quantum advantage in generative
models.
\newblock \emph{Physical Review Applied} 21, 044032 (2024).
\newblock doi:10.1103/PhysRevApplied.21.044032.

\bibitem{gu2012occam}
M. Gu, K. Wiesner, E. Rieper, and V. Vedral.
\newblock Quantum mechanics can reduce the complexity of classical models.
\newblock \emph{Nature Communications} 3, 762 (2012).

\bibitem{garner2017memory}
A. J. P. Garner, Q. Liu, J. Thompson, V. Vedral, and M. Gu.
\newblock Provably unbounded memory advantage in stochastic simulation using
quantum mechanics.
\newblock \emph{New Journal of Physics} 19, 103009 (2017).

\bibitem{kleinmann2011memory}
M. Kleinmann, O. G{\"u}hne, J. R. Portillo, J.-{\AA}. Larsson, and A. Cabello.
\newblock Memory cost of quantum contextuality.
\newblock \emph{New Journal of Physics} 13, 113011 (2011).

\bibitem{fagundes2017memory}
G. Fagundes and M. Kleinmann.
\newblock Memory cost for simulating all quantum correlations of the
Peres--Mermin scenario.
\newblock \emph{Journal of Physics A: Mathematical and Theoretical} 50,
325302 (2017).

\bibitem{karanjai2018contextuality}
A. Karanjai, J. J. Wallman, and S. D. Bartlett.
\newblock Contextuality bounds the efficiency of classical simulation of
quantum processes.
\newblock arXiv:1802.07744.

\bibitem{prakash2026memory}
S. Prakash.
\newblock Quantum memory advantage from contextuality.
\newblock arXiv:2607.00507, 2026.

\bibitem{li2025state}
B. Z. Li, Z. C. Guo, and J. Andreas.
\newblock (How) do language models track state?
\newblock In \emph{Proceedings of the 42nd International Conference on
Machine Learning}, PMLR 267, 34429--34452 (2025).

\bibitem{merrill2024illusion}
W. Merrill, J. Petty, and A. Sabharwal.
\newblock The illusion of state in state-space models.
\newblock In \emph{Proceedings of the 41st International Conference on
Machine Learning}, PMLR 235, 35492--35506 (2024).

\bibitem{peng2024limitations}
B. Peng, S. Narayanan, and C. Papadimitriou.
\newblock On limitations of the transformer architecture.
\newblock arXiv:2402.08164, 2024.

\bibitem{teo2025kcontextuality}
M. H. Teo, W. Yang, J. Sud, T. Tomesh, F. T. Chong, and E. R. Anschuetz.
\newblock \(k\)-Contextuality as a heuristic for memory separations in
learning.
\newblock In \emph{2025 IEEE International Conference on Quantum Computing
and Engineering (QCE)}, 2399 (2025).
\newblock doi:10.1109/QCE65121.2025.00260.

\bibitem{vaswani2017attention}
A. Vaswani et al.
\newblock Attention is all you need.
\newblock In \emph{Advances in Neural Information Processing Systems}, 2017.

\bibitem{nayak1999randomaccess}
A. Nayak.
\newblock Optimal lower bounds for quantum automata and random access codes.
\newblock In \emph{Proceedings of the 40th Annual Symposium on Foundations of
Computer Science}, 369--376 (1999).

\bibitem{fawzi2015psdrank}
H. Fawzi, J. Gouveia, P. A. Parrilo, R. Z. Robinson, and R. R. Thomas.
\newblock Positive semidefinite rank.
\newblock \emph{Mathematical Programming} 153, 133--177 (2015).

\bibitem{baryossef2004hiddenmatching}
Z. Bar-Yossef, T. S. Jayram, and I. Kerenidis.
\newblock Exponential separation of quantum and classical one-way communication
complexity.
\newblock In \emph{Proceedings of the 36th Annual ACM Symposium on Theory of
Computing}, 128--137 (2004).
\newblock doi:10.1145/1007352.1007379.

\bibitem{kerenidis2006booleanhiddenmatching}
I. Kerenidis and R. Raz.
\newblock The one-way communication complexity of the Boolean hidden matching
problem.
\newblock arXiv:quant-ph/0607173.

\bibitem{abramsky2011sheaf}
S. Abramsky and A. Brandenburger.
\newblock The sheaf-theoretic structure of non-locality and contextuality.
\newblock \emph{New Journal of Physics} 13, 113036 (2011).

\bibitem{abramsky2012databases}
S. Abramsky.
\newblock Relational databases and Bell's theorem.
\newblock arXiv:1208.6416.

\bibitem{kpv2024maxdicut}
J. Kallaugher, O. Parekh, and N. Voronova.
\newblock Exponential quantum space advantage for approximating maximum
directed cut in the streaming model.
\newblock In \emph{Proceedings of the 56th Annual ACM Symposium on Theory of
Computing}, 1805--1815 (2024).
\newblock doi:10.1145/3618260.3649709.

\bibitem{wang2026maxksat}
H. Wang and G. Yang.
\newblock Exponential quantum space advantage for approximating
Max-\(k\)SAT in the streaming setting.
\newblock arXiv:2606.05366, 2026.

\end{thebibliography}
\end{document}